\documentclass[11pt]{article}

\usepackage[utf8]{inputenc}
\usepackage[T1]{fontenc}

\usepackage{graphicx}
\usepackage{amsmath,amssymb,amsthm,amsfonts,mathtools}
\usepackage{hyperref}
\usepackage{enumitem}
\usepackage{booktabs}

\usepackage{natbib}

\usepackage{listings}
\usepackage{tikz}
\usepackage{pgfplots}
\pgfplotsset{compat=1.18}
\usetikzlibrary{arrows.meta,calc,positioning,decorations.pathmorphing}

\theoremstyle{definition}
\newtheorem{definition}{Definition}[section]
\newtheorem{assumption}{Assumption}[section]
\newtheorem{example}{Example}[section]

\theoremstyle{plain}
\newtheorem{theorem}{Theorem}[section]
\newtheorem{proposition}{Proposition}[section]
\newtheorem{lemma}{Lemma}[section]
\newtheorem{corollary}{Corollary}[section]

\theoremstyle{remark}
\newtheorem{remark}{Remark}[section]

\newcommand{\R}{\mathbb{R}}

\newcommand{\norm}[1]{\left\lVert #1\right\rVert}

\DeclareMathOperator{\Cov}{Cov}

\newcommand{\trS}{\operatorname{Tr}_{S}}                 
\newcommand{\trB}{\operatorname{Tr}_{\partial\Omega}}    

\title{\textbf{Geostatistics from Elliptic Boundary-Value Problems:}\\
\textbf{Green Operators, Transmission Conditions, and Schur Complements}}
\author{J.\ J.\ Segura\\
\small Universidad Andres Bello, Facultad de Ingeniería, Santiago, Chile\\
\small \texttt{juan.segura.f@unab.cl} \quad ORCID: 0009-0004-2742-4353}
\date{}

\begin{document}
\maketitle

\begin{abstract}
Classical geostatistics encodes spatial dependence by prescribing variograms or covariance kernels on Euclidean domains, whereas the SPDE--GMRF paradigm specifies Gaussian fields through an elliptic precision operator whose inverse is the corresponding Green operator.
We develop an operator-based formulation of Gaussian spatial random fields on bounded domains and manifolds with internal interfaces, treating boundary and transmission conditions as explicit components of the statistical model.
Starting from coercive quadratic energy functionals, variational theory yields a precise precision--covariance correspondence and shows that variograms are derived quadratic functionals of the Green operator, hence depend on boundary conditions and domain geometry.
Conditioning and kriging follow from standard Gaussian update identities in both covariance and precision form, with hard constraints represented equivalently by exact interpolation constraints or by distributional source terms.
Interfaces are modelled via surface penalty terms; taking variations produces flux-jump transmission conditions and induces controlled attenuation of cross-interface covariance.
Finally, boundary-driven prediction and domain reduction are formulated through Dirichlet-to-Neumann operators and Schur complements, providing an operator language for upscaling, change of support, and subdomain-to-boundary mappings.
Throughout, we use tools standard in spatial statistics and elliptic PDE theory to keep boundary and interface effects explicit in covariance modeling and prediction.
\end{abstract}


\section{Introduction}\label{sec:intro}

Spatial random fields are central in geostatistics and spatial statistics, where dependence is commonly described through covariance functions and variograms \citep{Matheron1963,Matheron1971,Cressie1993,ChilesDelfiner2012}.
A complementary operator-centric viewpoint defines Gaussian fields through \emph{precision operators}, often via elliptic SPDE models, so that covariances are Green operators of boundary-value problems \citep{Whittle1954,Whittle1967,LindgrenRueLindstrom2011,LindgrenBolinRue2022}.
This perspective is especially informative on bounded or irregular domains and in the presence of internal interfaces, because boundary conditions and transmission conditions enter the Green operator and therefore directly shape covariance, variograms, and kriging predictions.

\medskip
\noindent\textbf{Related work and positioning.}
Three strands are closest to the present article.
\emph{(i) SPDE--Mat\'ern and GMRF discretizations} make the operator--kernel duality computationally central by defining Mat\'ern fields as solutions of Whittle-type operators and exploiting sparse precisions \citep{Whittle1954,Whittle1967,LindgrenRueLindstrom2011,LindgrenBolinRue2022}.
\emph{(ii) Energy-based spatial random fields} (including Sparta/SRF constructions) build covariance models from Boltzmann--Gibbs energy functionals, emphasizing locality, statistical mechanics structure, and Green-function correlations \citep{HristopulosElogne2007,Hristopulos2020,AllardHristopulosOpitz2021,Hristopulos2022BGSPH}.
\emph{(iii) Gaussian free fields and QFT free fields} provide the mathematical template: quadratic actions on manifolds, Green operators, and (when needed) distributional kernels \citep{Cardy1996,Polyakov1987,DiFrancescoMathieuSenechal1997,Polchinski1998}.
The emphasis here is \emph{operator-first}: we keep boundary conditions and interfaces explicit at the level of the Green operator, and we place (a) surface-defect interfaces and their variational transmission conditions, (b) boundary effective operators (Dirichlet-to-Neumann) and exact domain reduction (Schur complements), and (c) the metric vs.\ deformation duality for non-stationarity, into a single narrative aligned with spatial statistical modeling and elliptic PDE theory on bounded domains.

\medskip
\noindent\textbf{Technical ingredients used.}
The rigorous layer uses only coercive variational theory (Lax--Milgram), Gaussian conditioning identities,
and exact operator reductions (Schur complements/DtN) for elliptic boundary-value problems.

\medskip
\noindent\textbf{Methodological commitment.}
We maintain two layers throughout:
\begin{itemize}[leftmargin=2em]
\item \textbf{Rigorous backbone:} coercive forms $\to$ Green operators (Lax--Milgram); Gaussian vectors/fields with precision $\leftrightarrow$ covariance; conditioning identities; Schur complements; DtN maps; interface transmission from surface terms; deformation pullbacks.
\item \textbf{Interpretative layer:} worldsheet/QFT vocabulary (actions, sources, defects/branes, bulk--boundary maps) used to organize modeling choices without importing physical string dynamics.
\end{itemize} 
\noindent\emph{Terminology note.} We use the terms \emph{energy}, \emph{source}, and \emph{defect} as shorthand for quadratic forms, linear functionals (observations/forcing), and interface penalty terms, respectively; no physical interpretation is assumed.

\medskip
\noindent\textbf{Structure.}
Section~\ref{sec:framework} sets the operator backbone and clarifies ``operator inverse'' vs ``integral kernel'' (Section~\ref{sec:kernel-vs-operator}).
Section~\ref{sec:polyakov} specializes to Polyakov-type quadratic actions and metric-induced anisotropy/non-stationarity.
Section~\ref{sec:worked} gives worked micro-examples showing boundary-condition and interface dependence of Green kernels/variograms.
Sections~\ref{sec:product}--\ref{sec:conditioning} treat nested structures (product manifolds), conditioning/kriging, and ``hard'' constraints as sources.
Section~\ref{sec:defects} formalizes defects/interfaces as surface actions.
Section~\ref{sec:deformation} treats deformation vs metric geometry for non-stationarity.
Sections~\ref{sec:cos}--\ref{sec:boundary} treat change of support, DtN maps, and Schur complements (domain reduction).
Section~\ref{sec:path} restates Gaussian identities in generating-functional notation and summarizes conditional simulation.
Section~\ref{sec:synthesis} summarizes the main dualities and provides implementation pseudocode.

\medskip
\noindent\textbf{Contributions (U1).}
The main contributions are structural and exact at the Gaussian/operator level:
\begin{itemize}[leftmargin=2em]
\item[(C1)] A unified coercive-form (operator-first) formulation for spatial Gaussian fields on bounded domains/manifolds in which \emph{boundary conditions are part of the covariance/variogram} through the Green operator of a boundary-value problem.
\item[(C2)] A defect/interface modelling device: a surface energy penalty $\frac{\alpha}{2}\int_S (\trS Z)^2\,dS$ yields a standard flux-jump transmission condition by variation and induces \emph{controlled cross-interface covariance attenuation}, illustrated by a worked micro-example.
\item[(C3)] A boundary-reduction viewpoint: Dirichlet-to-Neumann maps and discrete Schur complements are presented as the continuum/discrete forms of exact Gaussian elimination, providing a single operator language for upscaling, domain restriction, and bulk--boundary mappings.
\end{itemize}

\section{Operator Backbone: Actions, Green Operators, and Gaussian Fields}\label{sec:framework}

\subsection{Energy spaces and coercive forms}

\begin{assumption}[Prototypical elliptic setting]\label{ass:elliptic}
Let $\Omega\subset\R^d$ be a bounded Lipschitz domain.
Let $A:\Omega\to\R^{d\times d}$ be measurable, symmetric, and uniformly elliptic: there exist
$0<a_{\min}\le a_{\max}<\infty$ such that for a.e.\ $x\in\Omega$ and all $\xi\in\R^d$,
\[
a_{\min}\norm{\xi}^2 \le \xi^\top A(x)\xi \le a_{\max}\norm{\xi}^2.
\]
Let $c\in L^\infty(\Omega)$ satisfy $c(x)\ge c_0>0$ a.e.
\end{assumption}

\begin{definition}[Energy space and bilinear form]\label{def:energy}
Let $V:=H^1_0(\Omega)$ (Dirichlet boundary condition, for concreteness).
Define
\[
a(u,v)
:=\int_\Omega \nabla u(x)^\top A(x)\nabla v(x)\,dx
+\int_\Omega c(x)\,u(x)v(x)\,dx,
\qquad u,v\in V.
\]
\end{definition}

\begin{lemma}[Boundedness and coercivity]\label{lem:coercive}
Under Assumption~\ref{ass:elliptic}, there exist $M,\alpha>0$ such that
\[
|a(u,v)|\le M\norm{u}_{H^1(\Omega)}\norm{v}_{H^1(\Omega)},\qquad
a(u,u)\ge \alpha\norm{u}_{H^1(\Omega)}^2.
\]
\end{lemma}

\begin{proof}
Cauchy--Schwarz and the bounds on $A$ and $c$ yield boundedness. Coercivity follows from uniform ellipticity and $c\ge c_0$:
\[
a(u,u)\ge a_{\min}\norm{\nabla u}_{L^2}^2+c_0\norm{u}_{L^2}^2
\ge \min\{a_{\min},c_0\}\,\norm{u}_{H^1(\Omega)}^2.
\]
\end{proof}

\begin{theorem}[Lax--Milgram]\label{thm:laxmilgram}
Let $V$ be a Hilbert space and $a(\cdot,\cdot)$ a bounded coercive bilinear form on $V$.
Then for every $f\in V^\ast$ there exists a unique $u\in V$ such that
\[
a(u,v)=f(v)\qquad\forall v\in V,
\]
and $\norm{u}_V\le \alpha^{-1}\norm{f}_{V^\ast}$.
\end{theorem}

\begin{proof}
Define $J:V\to\R$ by $J(w):=\frac12\,a(w,w)-f(w)$.
Coercivity implies $J$ is strictly convex and coercive, hence admits a unique minimizer $u$.
The Euler--Lagrange condition $J'(u)[v]=0$ for all $v\in V$ yields $a(u,v)=f(v)$.
Taking $v=u$ gives the stability bound.
\end{proof}

\subsection{From bilinear form to operator and Green operator}

\begin{proposition}[Operator induced by a coercive form]\label{prop:form-operator}
Let $a(\cdot,\cdot)$ be bounded and coercive on a Hilbert space $V$.
Define $\mathcal{L}:V\to V^\ast$ by
\[
(\mathcal{L}u)(v) := a(u,v),\qquad u,v\in V.
\]
Then $\mathcal{L}$ is a bounded isomorphism $V\to V^\ast$. Moreover, for each $f\in V^\ast$,
the unique $u\in V$ solving $a(u,\cdot)=f(\cdot)$ satisfies $\mathcal{L}u=f$.
\end{proposition}

\begin{definition}[Green operator (variational inverse)]\label{def:green}
Under Theorem~\ref{thm:laxmilgram}, define $G:V^\ast\to V$ by $Gf=u$, where $u$ is the unique solution
$a(u,\cdot)=f(\cdot)$. Equivalently, $G=\mathcal{L}^{-1}$ in the sense of Proposition~\ref{prop:form-operator}.
We call $G$ the \emph{Green operator} associated with $(V,a)$.
\end{definition}

\begin{proposition}[Symmetry/positivity]\label{prop:greensym}
If $a$ is symmetric, then for all $f,g\in V^\ast$,
\[
f(Gg)=g(Gf),\qquad f(Gf)=a(Gf,Gf)\ge 0.
\]
\end{proposition}

\subsection{Discrete Gaussian fields and precision--covariance duality}

\begin{definition}[Gaussian vector with precision]\label{def:prec}
Let $Q\in\R^{n\times n}$ be symmetric positive definite.
A random vector $Z\in\R^n$ has \emph{precision} $Q$ if
\[
p(z)=\frac{\det(Q)^{1/2}}{(2\pi)^{n/2}}\exp\!\left(-\frac12 z^\top Q z\right),
\]
i.e.\ $Z\sim\mathcal{N}(0,Q^{-1})$.
\end{definition}

\begin{theorem}[Precision--covariance identity]\label{thm:preccov}
If $Z\sim\mathcal{N}(0,Q^{-1})$, then $\mathbb{E}[Z]=0$ and $\Cov(Z)=Q^{-1}$.
\end{theorem}

\subsection{Variograms as derived quadratic functionals}

\begin{definition}[Variogram]\label{def:variogram}
Let $Z$ be a second-order field on an index set $T$ with covariance $C(s,t)=\mathbb{E}[Z(s)Z(t)]$.
The (semi-)variogram is
\[
\gamma(s,t)=\frac12\,\mathbb{E}\big[(Z(s)-Z(t))^2\big],
\]
whenever finite.
\end{definition}

\begin{proposition}[Variogram--covariance identity]\label{prop:varcov}
If $Z$ is centered and second order, then
\[
\gamma(s,t)=\frac12\big(C(s,s)+C(t,t)-2C(s,t)\big).
\]
If moreover $C(s,s)$ is constant, then $\gamma(s,t)=C(s,s)-C(s,t)$.
\end{proposition}

% ============================================================
\subsection{Operator inverse vs.\ integral kernel (U3)}\label{sec:kernel-vs-operator}

A recurring ambiguity is ``$C(x,x')$ exists'' versus ``$G=\mathcal{L}^{-1}$ exists''. Our convention is:
\begin{itemize}[leftmargin=2em]
\item \textbf{Primary object:} the Green operator $G:V^\ast\to V$ defined variationally (Definition~\ref{def:green}).
This exists under coercivity and does \emph{not} require any pointwise kernel.
\item \textbf{Kernel representation (secondary):} when additional regularity holds, $G$ admits an integral kernel $G(x,x')$
(or a distributional kernel) such that $(Gf)(x)=\int_\Omega G(x,x')f(x')\,dx'$ for suitable $f$.
This depends on dimension, coefficients, and boundary regularity; see e.g.\ \citep{EvansPDE2010,GilbargTrudinger1983}.
\end{itemize}

\noindent\textbf{Terminology convention.}
We reserve \emph{Green operator} for the variational inverse $G=\mathcal{L}^{-1}:V^\ast\to V$.
When an integral representation exists we write \emph{Green's function} or \emph{Green kernel} for a function (or distribution)
$G(\cdot,\cdot)$ representing $G$ by $(Gf)(x)=\int_\Omega G(x,y)f(y)\,dy$ in the appropriate sense.

\begin{remark}[When we write $G(x,x')$]
Whenever we write $G(x,x')$ (or $C(x,x')$), this should be read as:
(i) an honest function kernel when regularity permits,
(ii) otherwise a \emph{distributional kernel} defined by duality.
All kriging/conditioning identities remain exact at the discrete level and meaningful at the operator level without requiring pointwise kernels.
\end{remark}

\begin{remark}[Variograms without diagonal kernels]
In intrinsic settings (e.g.\ massless fields in low dimension), $C(x,x)$ may fail pointwise.
Variograms remain meaningful through increments (Definition~\ref{def:variogram}) and/or discrete approximations
\citep{Matheron1971,Cardy1996}.
\end{remark}

% ============================================================
\section{Polyakov-Type Actions as Variogram Generators}\label{sec:polyakov}

Let $(\Sigma,h)$ be a 2D Riemannian manifold (possibly with boundary) with Laplace--Beltrami operator
\[
\Delta_h Z=\frac{1}{\sqrt{h}}\partial_a\!\left(\sqrt{h}\,h^{ab}\partial_b Z\right).
\]
(For the Polyakov/free-field viewpoint in physics, see \citep{Polyakov1987,Polchinski1998,DiFrancescoMathieuSenechal1997}.)

\begin{definition}[Polyakov-type quadratic action]\label{def:polyakov}
Let $\alpha'>0$ and $m\ge 0$.
Define
\[
S[Z]
=\frac{1}{4\pi\alpha'}\int_\Sigma \sqrt{h}\,h^{ab}\partial_a Z\,\partial_b Z\,d^2\sigma
+\frac{m^2}{2}\int_\Sigma \sqrt{h}\,Z^2\,d^2\sigma,
\]
with boundary conditions encoded in the energy space $V$ (Dirichlet/Neumann/Robin).
\end{definition}

\begin{proposition}[Associated precision operator]\label{prop:polyop}
Under boundary conditions compatible with integration by parts,
\[
S[Z]=\frac12\int_\Sigma \sqrt{h}\; Z\,(\mathcal{L}Z)\,d^2\sigma,
\qquad
\mathcal{L}=-\frac{1}{2\pi\alpha'}\Delta_h+m^2.
\]
\end{proposition}

\begin{remark}[Boundary-condition dependence is structural]
Even with identical differential expression $\mathcal{L}$, the Green operator depends on the boundary condition,
so the induced variogram does as well. Equivalently: the differential expression does not determine a unique
Green kernel without boundary/interface data.
Section~\ref{sec:worked} gives an explicit 1D calculation showing this dependence in closed form.
\end{remark}

% ============================================================
\section{Worked Micro-Examples: Boundaries and Interfaces (U2)}\label{sec:worked}

This section provides pocket calculations illustrating:
(i) Green kernels depend on boundary conditions, and
(ii) interface surface terms yield flux-jump conditions and attenuate cross-interface covariance.

% ------------------------------------------------------------
\subsection{1D interval: Dirichlet vs Neumann Green kernel and variogram}\label{sec:worked-1d}

Consider $\Omega=(0,L)$ and the 1D operator
\[
\mathcal{L}=-\frac{d^2}{dx^2}+m^2,\qquad m>0.
\]
Let $G_B(x,y)$ denote the Green kernel under boundary condition $B\in\{D,N\}$.

\begin{example}[Dirichlet Green kernel]\label{ex:dirichlet-green}
For Dirichlet $u(0)=u(L)=0$, one has
\[
G_D(x,y)=\frac{1}{m\sinh(mL)}
\begin{cases}
\sinh(mx)\,\sinh(m(L-y)), & x\le y,\\[2pt]
\sinh(my)\,\sinh(m(L-x)), & x>y.
\end{cases}
\]
The induced variogram is $\gamma_D(x,y)=\tfrac12(G_D(x,x)+G_D(y,y)-2G_D(x,y))$ when interpreted pointwise.
\end{example}

\begin{example}[Neumann Green kernel]\label{ex:neumann-green}
For Neumann $u'(0)=u'(L)=0$, the Green kernel is
\[
G_N(x,y)=\frac{1}{m\sinh(mL)}
\begin{cases}
\cosh(mx)\,\cosh(m(L-y)), & x\le y,\\[2pt]
\cosh(my)\,\cosh(m(L-x)), & x>y,
\end{cases}
\]
and $\gamma_N(x,y)=\tfrac12(G_N(x,x)+G_N(y,y)-2G_N(x,y))$.
\end{example}

\begin{remark}[Same operator, different variograms]
$G_D\neq G_N$ even though the differential operator is identical; only the boundary condition changes.
Therefore the induced variogram is boundary-condition dependent.
\end{remark}

\paragraph{Implication for empirical variography on bounded domains (U4).}
Examples~\ref{ex:dirichlet-green}--\ref{ex:neumann-green} and the interface model in Section~\ref{sec:worked-interface}
show that in an operator-first formulation the covariance $C=\mathcal{L}^{-1}$ and variogram $\gamma$ are determined by the
\emph{boundary-value problem}, not the differential expression alone: the domain $\Omega$, boundary condition $B$ on $\partial\Omega$,
and internal defects/transmission terms enter the Green operator. Consequently, experimental variograms computed from data supported
in a bounded region need not match the infinite-domain template associated with the same symbol (e.g.\ Mat\'ern/Whittle), even when
interior coefficients are constant: boundary screening/reflection and interface transmission modify long-lag behaviour and can shift
apparent sill and range.

A practical recommendation is to treat $(\mathcal{L},B,S)$ as the calibrated object: estimate parameters against the domain-dependent
Green operator (e.g.\ via SPDE/GMRF likelihood), and stratify empirical variogram diagnostics by pair-classes (interior vs near-boundary;
same-side vs cross-interface). A simple diagnostic is to compute empirical variograms separately for points within a buffer of $\partial\Omega$
versus deep interior; competing boundary conditions (Dirichlet/Neumann/Robin, or mixed) can be compared by likelihood or predictive score
within the same SPDE discretisation.

% ------------------------------------------------------------
\subsection{Two-subdomain interface: surface term and derivative jump}\label{sec:worked-interface}

Let $\Omega=(-L,L)$ with an interface at $S=\{0\}$.
Consider the energy (1D for simplicity)
\[
\mathcal{E}[Z]
=\frac12\int_{-L}^0 \big(|Z'|^2+m^2 Z^2\big)\,dx
+\frac12\int_0^{L} \big(|Z'|^2+m^2 Z^2\big)\,dx
+\frac{\alpha}{2}\,Z(0)^2,
\qquad \alpha\ge 0,
\]
with (say) Dirichlet at $\pm L$ to avoid endpoint terms.

\noindent\emph{Trace remark.} In one dimension the point trace $Z(0)$ is well-defined for $Z\in H^1(-L,L)$; in higher dimensions replace point evaluation by a surface trace on an interface hypersurface $S$ (cf.\ Proposition~\ref{prop:jump}).

\begin{proposition}[Interface (jump) condition]\label{prop:1d-jump}
Stationarity of $\mathcal{E}$ implies
\[
(-Z''+m^2 Z)=0\quad \text{on }(-L,0)\cup(0,L),
\]
continuity of $Z$ at $0$ (for $Z\in H^1$), and the derivative jump condition
\[
Z'(0^+)-Z'(0^-)=\alpha\,Z(0).
\]
\end{proposition}

\begin{remark}[Explicit Green-kernel update (delta barrier)]
Let $G_0$ be the Dirichlet Green kernel on $(-L,L)$ for $-d^2/dx^2+m^2$ without the interface penalty ($\alpha=0$).
Then
\[
G_\alpha(x,y)
=
G_0(x,y)\;-\;\frac{\alpha\,G_0(x,0)\,G_0(0,y)}{1+\alpha\,G_0(0,0)}.
\]
In particular $G_\alpha(x,y)\le G_0(x,y)$ pointwise, and cross-interface covariance ($x<0<y$) decreases monotonically with $\alpha$.
For the Dirichlet kernel on $(-L,L)$ one can evaluate $G_0(0,0)=\frac{\tanh(mL)}{2m}$.
\end{remark}

\begin{figure}[t]
  \centering
  \IfFileExists{effects.png}{%
    \includegraphics[width=\linewidth]{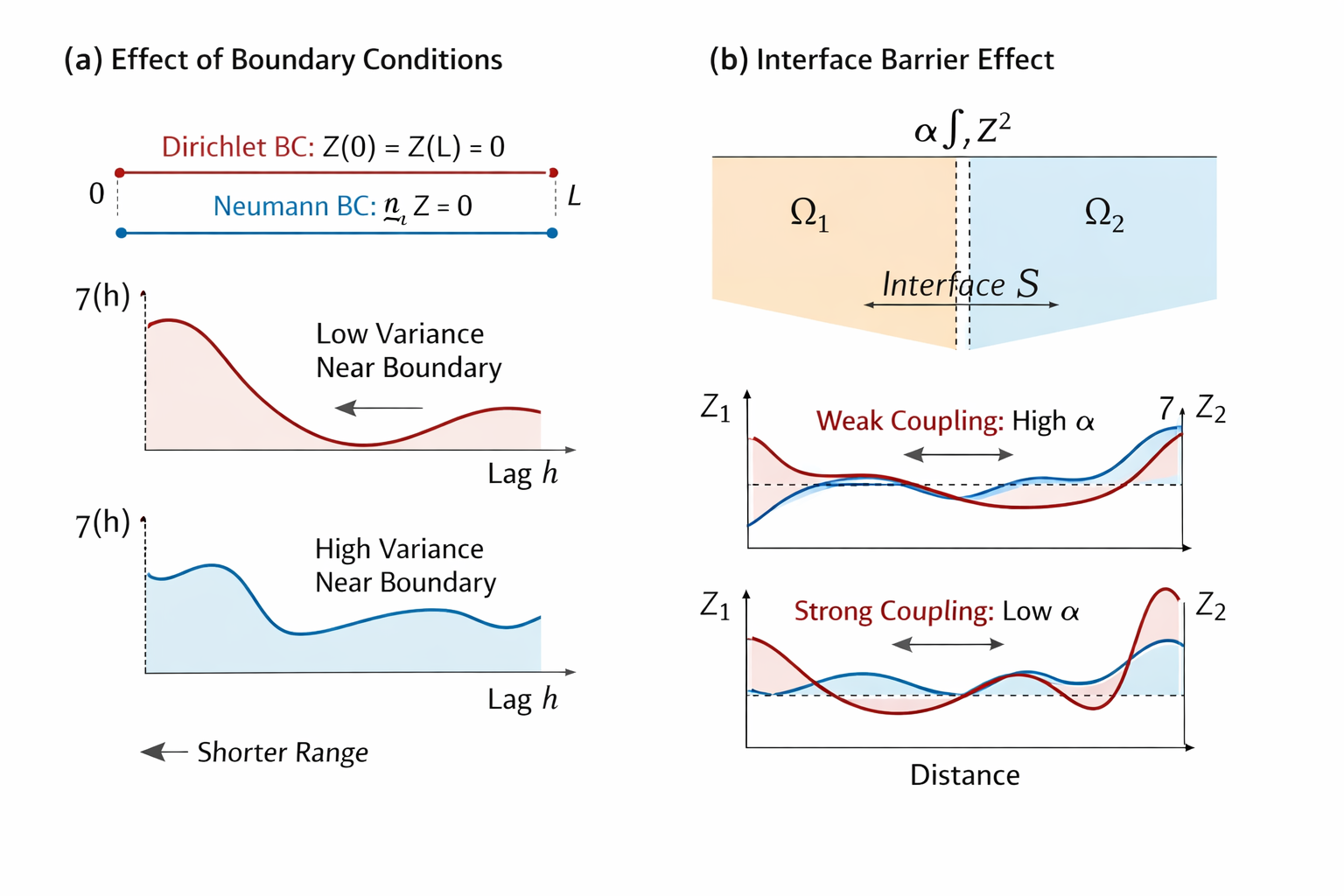}%
  }{%
    \fbox{\parbox{0.95\linewidth}{\centering\vspace{1em}
    \textbf{Figure placeholder.} File \texttt{effects.png} not found.\vspace{1em}}}%
  }
  \caption{Schematic domain effects on spatial dependence. \textbf{Left:} the same elliptic operator on a bounded domain yields different Green kernels (hence covariances/variograms) under Dirichlet vs.\ Neumann boundary conditions, illustrating boundary-induced screening/reflection. \textbf{Right:} an internal interface $S$ with defect penalty $\frac{\alpha}{2}\int_S Z^2\,dS$ (or equivalent transmission condition) reduces cross-interface correlation as $\alpha$ increases, while preserving stronger along-interface continuity.}
  \label{fig:effects}
\end{figure}

\section{Product Manifolds and Nested Covariance Structures}\label{sec:product}

Nested/multi-range covariance models are ubiquitous in practice \citep{ChilesDelfiner2012}.
A clean operator mechanism is to place the field on a product $\Sigma\times\mathcal{K}$ and decompose into internal modes.

\begin{assumption}[Compact internal factor]\label{ass:compact}
Let $\mathcal{K}$ be a compact Riemannian manifold without boundary.
Let $\{-\Delta_{\mathcal{K}}\phi_n=\lambda_n\phi_n\}_{n\ge 0}$ be an orthonormal eigenbasis of $L^2(\mathcal{K})$.
\end{assumption}

\begin{proposition}[Mode decomposition on $\Sigma\times\mathcal{K}$]\label{prop:kk}
Let $\mathcal{M}=\Sigma\times\mathcal{K}$ and
\[
\mathcal{L}=\mathcal{L}_\Sigma\otimes I + I\otimes(-\Delta_{\mathcal{K}})+m^2,
\]
with $\mathcal{L}_\Sigma$ positive self-adjoint (boundary conditions included).
If $G_n:=(\mathcal{L}_\Sigma+(m^2+\lambda_n))^{-1}$ exists (variationally) for each $n$, then
for $f(\sigma,y)=\sum_n f_n(\sigma)\phi_n(y)$,
\[
(\mathcal{L}^{-1}f)(\sigma,y)=\sum_{n\ge 0} (G_n f_n)(\sigma)\,\phi_n(y).
\]
\end{proposition}

\begin{remark}
Each internal mode shifts the effective mass $m_n^2=m^2+\lambda_n$, contributing a component with its own range.
If $\mathcal{L}_\Sigma=(\kappa^2-\Delta_\Sigma)^{\alpha/2}$ (Whittle--Mat\'ern form), then each mode produces
$\kappa_n^2=\kappa^2+\lambda_n$ and an effective range scale on the order of $(\kappa^2+\lambda_n)^{-1/2}$.
\end{remark}

% ============================================================
\section{Conditioning, Kriging, and Source Representations}\label{sec:conditioning}

\subsection{Gaussian conditioning in covariance form}

\begin{theorem}[Conditional law of a joint Gaussian vector]\label{thm:condgauss}
Let $(X,Y)$ be jointly Gaussian with mean zero and block covariance
\[
\begin{pmatrix}X\\Y\end{pmatrix}\sim
\mathcal{N}\!\left(0,
\begin{pmatrix}
\Sigma_{XX} & \Sigma_{XY}\\
\Sigma_{YX} & \Sigma_{YY}
\end{pmatrix}
\right),
\]
where $\Sigma_{XX}$ is symmetric positive definite (hence invertible). Then the conditional distribution of $Y$ given $X=x$ is Gaussian:
\[
Y\,\big|\,X=x \sim \mathcal{N}\!\left(\ \Sigma_{YX}\Sigma_{XX}^{-1}x,\ \Sigma_{YY}-\Sigma_{YX}\Sigma_{XX}^{-1}\Sigma_{XY}\ \right).
\]
The conditional covariance is the (positive semidefinite) Schur complement of $\Sigma_{XX}$ in the joint covariance matrix.
\end{theorem}

\subsection{Linear-Gaussian update in precision form}

\begin{theorem}[Linear-Gaussian update]\label{thm:linupdate}
Let $Z\sim\mathcal{N}(0,Q^{-1})$ with $Q$ SPD and observations
\[
z=RZ+\varepsilon,\qquad \varepsilon\sim\mathcal{N}(0,N),
\]
with $N$ SPD. Then
\[
Z\mid z\sim\mathcal{N}(\bar z,Q_{\mathrm{post}}^{-1}),\qquad
Q_{\mathrm{post}}=Q+R^\top N^{-1}R,\qquad
\bar z=Q_{\mathrm{post}}^{-1}R^\top N^{-1}z.
\]
\end{theorem}

\subsection{Hard conditioning as delta constraints / point sources}

\begin{proposition}[Delta constraints as Fourier representations]\label{prop:vertex}
Formally, hard constraints $Z(x_i)=z_i$ satisfy
\[
\prod_{i=1}^n\delta(Z(x_i)-z_i)
=
\int_{\R^n}\prod_{i=1}^n\frac{d\lambda_i}{2\pi}\,
\exp\!\left(i\sum_{i=1}^n\lambda_i Z(x_i)-i\sum_{i=1}^n\lambda_i z_i\right).
\]
\end{proposition}

\begin{remark}
In operator language, conditioning inserts linear functionals (rows of $R$). In interpretative field-theory language, the Fourier factors act as point sources.
\end{remark}

\subsection{Parameter inference and identifiability}\label{sec:inference}

Let $\theta$ collect hyperparameters in the prior precision operator and defect terms, e.g.
$\theta=(m,\alpha,\text{metric/anisotropy parameters},\ldots)$.
With linear observations $z=RZ+\varepsilon$, $\varepsilon\sim\mathcal{N}(0,N)$, the marginal likelihood is
\[
z \sim \mathcal{N}\bigl(0,\; \Sigma_\theta\bigr), \qquad
\Sigma_\theta = R\,\mathcal{L}_\theta^{-1}R^\top + N,
\]
and
\[
\ell(\theta) = -\tfrac12 z^\top \Sigma_\theta^{-1}z - \tfrac12\log\det(\Sigma_\theta) + \text{const}.
\]
This is the same likelihood optimized in SPDE--GMRF practice; in particular, it applies to defect/interface parameters such as $\alpha$
and to metric parameters controlling anisotropy \citep{LindgrenRueLindstrom2011,LindgrenBolinRue2022}.

% ============================================================
\section{Defects and Interfaces via Surface Terms}\label{sec:defects}

\begin{proposition}[Surface penalty induces transmission condition]\label{prop:jump}
Let $\Omega=\Omega_1\cup\Omega_2$ with smooth interface $S=\partial\Omega_1\cap\partial\Omega_2$.
Consider
\[
\mathcal{E}[Z]
=\frac12\sum_{k=1}^2\int_{\Omega_k}\big(\nabla Z^\top A_k\nabla Z+c_k Z^2\big)\,dx
+\frac{\alpha}{2}\int_S (\trS Z)^2\,dS.
\]
Stationarity under variations supported on $S$ yields the interface condition
\[
(A_1\nabla Z_1)\cdot n_1+(A_2\nabla Z_2)\cdot n_2+\alpha\,\trS Z=0.
\]
\end{proposition}

\begin{remark}[Continuity vs.\ flux jump]\label{rem:cont-vs-flux}
In this single-field $H^{1}(\Omega)$ setting the trace of $Z$ is single-valued across $S$ (no jump in $Z$);
the defect term produces a jump only in the normal flux, as made explicit in Remark~\ref{rem:jumpform}.
\end{remark}

\begin{remark}[Flux jump form and trace regularity]\label{rem:jumpform}
Let $n:=n_1=-n_2$ be a chosen unit normal along $S$ (so the outward normals satisfy $n_2=-n_1$).
Then the interface condition in Proposition~\ref{prop:jump} can be written in the standard jump form
\[
\big[(A\nabla Z)\cdot n\big]_{S}+\alpha\,\trS Z=0,
\]
where $\big[w\big]_{S}:=w|_{\Omega_1}-w|_{\Omega_2}$ denotes the jump across $S$ and $\trS Z$ is the Sobolev trace of $Z$ on $S$.
In the $H^1$-based setting, $\trS Z\in H^{1/2}(S)$ is well-defined, so the surface penalty
$\frac{\alpha}{2}\int_S (\trS Z)^2\,dS$ is meaningful. If discontinuities of $Z$ are intended, one can instead penalize the trace jump,
e.g.\ $\int_S |\trS(Z_1-Z_2)|^2\,dS$.
\end{remark}

% ============================================================
\section{Non-stationarity: Deformation vs Metric Geometry}\label{sec:deformation}

\begin{definition}[Pullback of a random field]\label{def:pullback}
Let $W$ be a random field on $\Omega'$ and $f:\Omega\to\Omega'$ measurable.
Define $Z(x)=W(f(x))$.
\end{definition}

\begin{proposition}[Covariance under pullback]\label{prop:pullcov}
If $W$ is second order with covariance $C_W$, then $Z$ has covariance
\[
C_Z(x,x')=C_W(f(x),f(x')).
\]
If $W$ is stationary on $\Omega'$, then $C_Z(x,x')=C_0(f(x)-f(x'))$, typically non-stationary on $\Omega$.
\end{proposition}

\begin{remark}
In 2D, conformal maps induce conformal metrics and can be used to encode spatial deformation \citep{SampsonGuttorp1992,ChilesDelfiner2012}.
\end{remark}

% ============================================================
\section{Change of Support and Spectral Identities}\label{sec:cos}

\begin{proposition}[Linear images of Gaussians]\label{prop:linimage}
If $Z\sim\mathcal{N}(0,C)$ and $Y=AZ$, then $Y\sim\mathcal{N}(0,ACA^\top)$.
\end{proposition}

\begin{proposition}[Spectral effect of averaging filters]\label{prop:filter}
If $Z$ is stationary with spectral density $f_Z(k)$ and $Y=h*Z$, then
\[
f_Y(k)=|\widehat{h}(k)|^2\,f_Z(k).
\]
\end{proposition}

\begin{remark}
Coarse supports suppress high-frequency modes; refitting within a parametric family induces a parameter flow.
\end{remark}

% ============================================================
\section{Boundaries, DtN Maps, and Schur Complements}\label{sec:boundary}

\subsection{Dirichlet-to-Neumann (DtN) map}

\begin{definition}[Dirichlet-to-Neumann map]\label{def:dtn}
Let $\Omega\subset\R^d$ be a bounded Lipschitz domain and let
\[
\mathcal{L}u := -\nabla\cdot(A\nabla u)+cu
\]
with $A$ bounded, measurable and uniformly elliptic, and $c\in L^\infty(\Omega)$ with $c\ge 0$.
For $\varphi\in H^{1/2}(\partial\Omega)$ choose any lifting $w_\varphi\in H^1(\Omega)$ such that
$\trB w_\varphi=\varphi$, and let $v_\varphi\in H^1_0(\Omega)$ be the unique weak solution of
\[
\int_\Omega \nabla v_\varphi^\top A\nabla \eta\,dx+\int_\Omega c\,v_\varphi\,\eta\,dx
=
-\int_\Omega \nabla w_\varphi^\top A\nabla \eta\,dx-\int_\Omega c\,w_\varphi\,\eta\,dx,
\qquad \forall \eta\in H^1_0(\Omega).
\]
Set $u_\varphi:=w_\varphi+v_\varphi\in H^1(\Omega)$; then $\mathcal{L}u_\varphi=0$ in the weak sense and
$\trB u_\varphi=\varphi$.

The Dirichlet-to-Neumann operator is the bounded map $\Lambda:H^{1/2}(\partial\Omega)\to H^{-1/2}(\partial\Omega)$ defined by
\[
\langle \Lambda\varphi,\psi\rangle_{H^{-1/2},H^{1/2}}
:=\int_\Omega \nabla u_\varphi^\top A\nabla u_\psi\,dx+\int_\Omega c\,u_\varphi\,u_\psi\,dx,
\qquad \forall \psi\in H^{1/2}(\partial\Omega),
\]
where $u_\psi$ is constructed as above. Equivalently, $\Lambda\varphi$ is the conormal flux
$(A\nabla u_\varphi)\cdot n$ on $\partial\Omega$ understood as an element of $H^{-1/2}(\partial\Omega)$.
\end{definition}

\begin{remark}[Notation]
When $u_\varphi$ is regular enough, the duality pairing
$\langle \Lambda\varphi,\varphi\rangle_{H^{-1/2},H^{1/2}}$ coincides with the boundary integral
$\int_{\partial\Omega}\varphi\,(\Lambda\varphi)\,dS$.
\end{remark}

\begin{theorem}[DtN energy identity]\label{thm:dtn}
Assuming sufficient regularity (or interpreting the right-hand side as a duality pairing),
\[
\int_\Omega \nabla u_\varphi^\top A\nabla u_\varphi\,dx+\int_\Omega c\,u_\varphi^2\,dx
=\langle \Lambda\varphi,\varphi\rangle_{H^{-1/2},H^{1/2}}.
\]
\end{theorem}

\subsection{Schur complement reduction}

\begin{lemma}[Schur complement as marginal precision]\label{lem:schur}
Let $(U,V)$ be jointly Gaussian with block precision
\[
Q=\begin{pmatrix}
Q_{UU}&Q_{UV}\\
Q_{VU}&Q_{VV}
\end{pmatrix},\qquad Q\succ 0.
\]
Then the marginal law of $U$ has precision
\[
Q_{\mathrm{marg}}=Q_{UU}-Q_{UV}Q_{VV}^{-1}Q_{VU}.
\]
\end{lemma}

\begin{remark}
DtN operators are continuum boundary effective operators; Schur complements are their discrete analogues.
\end{remark}

% ============================================================
\section{Generating Functionals and Conditional Simulation}\label{sec:path}

\begin{theorem}[Gaussian generating functional]\label{thm:genfun}
If $Z\sim\mathcal{N}(0,Q^{-1})$, then for any $J$,
\[
\mathbb{E}\big[e^{J^\top Z}\big]=\exp\!\left(\tfrac12 J^\top Q^{-1}J\right).
\]
\end{theorem}

\begin{lemma}[Sampling from a precision]\label{lem:sampling}
Let $Q=LL^\top$ be a Cholesky factorization and $\xi\sim\mathcal{N}(0,I)$.
If $L^\top \eta=\xi$, then $\eta\sim\mathcal{N}(0,Q^{-1})$.
\end{lemma}

\begin{corollary}[Posterior conditional simulation]\label{cor:postsamp}
In Theorem~\ref{thm:linupdate}, a conditional simulation is
\[
Z^{(s)}=\bar z+\eta,\qquad \eta\sim\mathcal{N}(0,Q_{\mathrm{post}}^{-1}).
\]
\end{corollary}

% ============================================================
\section{Multivariate fields and matrix-valued operators}\label{sec:multivariate}

Multivariate Gaussian fields can be specified by block (matrix-valued) precision operators; after discretisation this yields an SPD block matrix $Q$ and cross-covariances from $C=Q^{-1}$.
A simple separable (coregionalization) case is
\[
Q = G^{-1}\otimes Q_0 \quad\Rightarrow\quad C=G\otimes Q_0^{-1},
\]
recovering a shared spatial kernel with variable-to-variable coupling \citep{ChilesDelfiner2012}.

% ============================================================
\section{Synthesis: Dualities and Implementation Patterns}\label{sec:synthesis}

\subsection{Three Gaussian-level dualities}
\begin{itemize}[leftmargin=2em]
\item \textbf{Operator vs kernel:} specifying $Q$ is equivalent to specifying $C=Q^{-1}$ (discretely), and variationally $G=\mathcal{L}^{-1}$ is the primary object (Section~\ref{sec:kernel-vs-operator}).
\item \textbf{Metric vs deformation:} non-stationarity can be encoded either by variable coefficients/metrics or by pullback of a stationary field (Section~\ref{sec:deformation}) \citep{SampsonGuttorp1992}.
\item \textbf{Bulk vs boundary:} DtN maps and Schur complements encode how eliminated bulk degrees of freedom induce boundary/subdomain effective operators (Section~\ref{sec:boundary}).
\end{itemize}

\subsection{Minimal implementation pseudocode}

\paragraph{(A) FEM/SPDE assembly + kriging + conditional simulation (precision form).}
\begin{lstlisting}
Input: mesh/graph, operator parameters (A(x), c(x), BCs), data locations/rows R, noise N
1) Assemble sparse precision Q from bilinear form a(u,v) (FEM stiffness + mass terms),
   including boundary conditions and (optional) interface surface terms.
2) Posterior precision: Q_post = Q + R^T N^{-1} R
3) Posterior mean: solve Q_post * zbar = R^T N^{-1} z
4) Conditional simulation:
   - sample g ~ N(0,I)
   - compute eta = Q_post^{-1/2} g  (exact via Cholesky; approximate via Lanczos/rational)
   - sample = zbar + eta
Output: zbar, samples
\end{lstlisting}

\paragraph{(B) Domain reduction by Schur complement.}
\begin{lstlisting}
Partition nodes into keep (K) and eliminate (E):
Q = [[Q_KK, Q_KE],
     [Q_EK, Q_EE]]
Effective precision on K:
Q_eff = Q_KK - Q_KE * (Q_EE^{-1}) * Q_EK
Use Q_eff for inference restricted to subdomain K.
\end{lstlisting}

\paragraph{(C) Boundary-condition sensitivity check.}
\begin{lstlisting}
On a 1D interval (0,L), build Q_D and Q_N enforcing Dirichlet vs Neumann BCs.
Compute C_D = Q_D^{-1}, C_N = Q_N^{-1}.
Compare variograms gamma_D(i,j)=0.5*(C_D(ii)+C_D(jj)-2*C_D(ij)) and gamma_N analogously.
\end{lstlisting}

% ============================================================


\begin{thebibliography}{99}

\bibitem[Allard et~al.(2021)Allard, Hristopulos, and Opitz]{AllardHristopulosOpitz2021}
Denis Allard, Dionissios~T. Hristopulos, and Thomas Opitz.
\newblock Linking physics and spatial statistics: A new family of {B}oltzmann--{G}ibbs random fields.
\newblock {\em Electronic Journal of Statistics}, 15(2):4085--4116, 2021.

\bibitem[Cardy(1996)]{Cardy1996}
John Cardy.
\newblock {\em Scaling and Renormalization in Statistical Physics}.
\newblock Cambridge University Press, 1996.

\bibitem[Chil{\`e}s and Delfiner(2012)]{ChilesDelfiner2012}
Jean-Paul Chil{\`e}s and Pierre Delfiner.
\newblock {\em Geostatistics: Modeling Spatial Uncertainty}.
\newblock Wiley, 2nd edition, 2012.

\bibitem[Cressie(1993)]{Cressie1993}
Noel Cressie.
\newblock {\em Statistics for Spatial Data}.
\newblock Wiley, revised edition, 1993.

\bibitem[Di~Francesco et~al.(1997)Di~Francesco, Mathieu, and S{\'e}n{\'e}chal]{DiFrancescoMathieuSenechal1997}
Philippe Di~Francesco, Pierre Mathieu, and David S{\'e}n{\'e}chal.
\newblock {\em Conformal Field Theory}.
\newblock Springer, 1997.

\bibitem[Evans(2010)]{EvansPDE2010}
Lawrence C. Evans.
\newblock {\em Partial Differential Equations}.
\newblock Graduate Studies in Mathematics, Vol.~19. American Mathematical Society, 2nd edition, 2010.

\bibitem[Gilbarg and Trudinger(1983)]{GilbargTrudinger1983}
David Gilbarg and Neil S. Trudinger.
\newblock {\em Elliptic Partial Differential Equations of Second Order}.
\newblock Springer, 2nd edition, 1983.

\bibitem[Hristopulos(2020)]{Hristopulos2020}
Dionissios~T. Hristopulos.
\newblock {\em Random Fields for Spatial Data Modeling: A Primer for Scientists and Engineers}.
\newblock Springer, 2020.

\bibitem[Hristopulos(2022)]{Hristopulos2022BGSPH}
Dionissios~T. Hristopulos.
\newblock {B}oltzmann--{G}ibbs random fields with mesh-free precision operators based on smoothed particle hydrodynamics.
\newblock {\em Theory of Probability and Mathematical Statistics}, 107:37--60, 2022.

\bibitem[Hristopulos and Elogne(2007)]{HristopulosElogne2007}
Dionissios~T. Hristopulos and Steeve~N. Elogne.
\newblock Analytic properties and covariance functions of a new class of generalized {G}ibbs random fields.
\newblock {\em IEEE Transactions on Information Theory}, 53(12):4667--4679, 2007.

\bibitem[Lindgren et~al.(2011)Lindgren, Rue, and Lindstr{\"o}m]{LindgrenRueLindstrom2011}
Finn Lindgren, H{\aa}vard Rue, and Johan Lindstr{\"o}m.
\newblock An explicit link between {G}aussian fields and {G}aussian {M}arkov random fields: The {SPDE} approach.
\newblock {\em Journal of the Royal Statistical Society: Series B}, 73(4):423--498, 2011.

\bibitem[Lindgren et~al.(2022)Lindgren, Bolin, and Rue]{LindgrenBolinRue2022}
Finn Lindgren, David Bolin, and H{\aa}vard Rue.
\newblock The {SPDE} approach for {G}aussian and non-{G}aussian fields: 10 years and still running.
\newblock {\em Spatial Statistics}, 50:100599, 2022.

\bibitem[Matheron(1963)]{Matheron1963}
Georges Matheron.
\newblock Principles of geostatistics.
\newblock {\em Economic Geology}, 58(8):1246--1266, 1963.

\bibitem[Matheron(1971)]{Matheron1971}
Georges Matheron.
\newblock {\em The Theory of Regionalized Variables and Its Applications}.
\newblock \'Ecole Nationale Sup\'erieure des Mines de Paris, 1971.

\bibitem[Polchinski(1998)]{Polchinski1998}
Joseph Polchinski.
\newblock {\em String Theory}, Vols.\ 1--2.
\newblock Cambridge University Press, 1998.

\bibitem[Polyakov(1987)]{Polyakov1987}
Alexander~M. Polyakov.
\newblock {\em Gauge Fields and Strings}.
\newblock Harwood Academic Publishers, 1987.

\bibitem[Sampson and Guttorp(1992)]{SampsonGuttorp1992}
Paul~D. Sampson and Peter Guttorp.
\newblock Nonparametric estimation of nonstationary spatial covariance structure.
\newblock {\em Journal of the American Statistical Association}, 87(417):108--119, 1992.

\bibitem[Whittle(1954)]{Whittle1954}
Peter Whittle.
\newblock On stationary processes in the plane.
\newblock {\em Biometrika}, 41(3--4):434--449, 1954.

\bibitem[Whittle(1967)]{Whittle1967}
Peter Whittle.
\newblock Stochastic-processes in several dimensions.
\newblock In J.\ Neyman (ed.), {\em Proc.\ Fifth Berkeley Symp.\ Math.\ Stat.\ Prob.\ Vol.\ 2}, 305--314.
\newblock University of California Press, 1967.

\end{thebibliography}
\end{document}